\documentclass[11pt]{article}
\usepackage[letterpaper,margin=1in]{geometry}
\usepackage{amsmath, amssymb, amsthm, amsfonts}
\usepackage{cite}

\usepackage{bbm}

\usepackage{ifthen}
\usepackage{tikz}
\usetikzlibrary{positioning,decorations.pathreplacing}

\usepackage{appendix}
\usepackage{graphicx}
\usepackage{color}
\usepackage{algorithm}
\usepackage[noend]{algpseudocode}
\usepackage{epstopdf}
\usepackage{wrapfig}
\usepackage{thm-restate}

\usepackage{wasysym}
\usepackage[textsize=tiny]{todonotes}

\usepackage{framed}
\usepackage[framemethod=tikz]{mdframed}
\usepackage[bottom]{footmisc}
\usepackage{enumitem}
\setitemize{noitemsep,topsep=3pt,parsep=3pt,partopsep=3pt}
\usepackage[font=small]{caption}
\usepackage{xspace}

\newtheorem{theorem}{Theorem}[section]
\newtheorem{lemma}[theorem]{Lemma}
\newtheorem{meta-theorem}[theorem]{Meta-Theorem}
\newtheorem{claim}[theorem]{Claim}

\newtheorem{definition}[theorem]{Definition}
\newtheorem{invariant}[theorem]{Invariant}

\definecolor{darkgreen}{rgb}{0,0.5,0}
\usepackage{hyperref}
\hypersetup{
    unicode=false,          
    colorlinks=true,        
    linkcolor=red,          
    citecolor=darkgreen,        
    filecolor=magenta,      
    urlcolor=cyan           
}\usepackage[capitalize, nameinlink]{cleveref}
\Crefname{lemma}{Lemma}{Lemmas}
\Crefname{claim}{Claim}{Claims}
\Crefname{remark}{Remark}{Remarks}
\Crefname{observation}{Observation}{Observations}
\Crefname{invariant}{Invariant}{Invariants}
\algnewcommand\algorithmicswitch{\textbf{switch}}
\algnewcommand\algorithmiccase{\textbf{case}}

\algdef{SE}[SWITCH]{Switch}{EndSwitch}[1]{\algorithmicswitch\ #1\ \algorithmicdo}{\algorithmicend\ \algorithmicswitch}%
\algdef{SE}[CASE]{Case}{EndCase}[1]{\algorithmiccase\ #1}{\algorithmicend\ \algorithmiccase}%
\algtext*{EndSwitch}%
\algtext*{EndCase}%

\newcommand{\eps}{\varepsilon}

\newcommand{\congest}{$\mathsf{CONGEST}$\xspace}

\newcommand{\poly}{\operatorname{\text{{\rm poly}}}}

\renewcommand{\paragraph}[1]{\medskip\noindent {\bf #1}}

\let\oldtextbf=\textbf
\renewcommand\textbf[1]{{\boldmath\oldtextbf{#1}}}

\newcommand{\FullOrShort}{full}

\ifthenelse{\equal{\FullOrShort}{full}}{
	
  \newcommand{\fullOnly}[1]{#1}
  \newcommand{\shortOnly}[1]{}
}{
    \newcommand{\shortOnly}[1]{#1}
    \newcommand{\fullOnly}[1]{}
  }

\begin{document}
\date{}
\title{Parallel Dynamic Maximal Matching}

\author{
  Mohsen Ghaffari\\
  \small MIT \\
  \small ghaffari@mit.edu
  \and 
  Anton Trygub\\
  \small MIT \\
  \small trygub@mit.edu
 }

\maketitle
\begin{abstract} 
We present the first (randomized) parallel dynamic algorithm for maximal matching, which can process an arbitrary number of updates simultaneously. Given a batch of edge deletion or insertion updates to the graph, our parallel algorithm adjusts the maximal matching to these updates in $\poly(\log n)$ depth and using $\poly(\log n)$ amortized work per update. That is, the amortized work for processing a batch of $k$ updates is $k\poly(\log n)$, while all this work is done in $\poly(\log n)$ depth, with high probability. This can be seen as a parallel counterpart of the sequential dynamic algorithms for constant-approximate and maximal matching [Onak and Rubinfeld STOC'10; Baswana, Gupta, and Sen FOCS'11; and Solomon FOCS'16]. 
Our algorithm readily generalizes to maximal matching in hypergraphs of rank $r$---where each hyperedge has at most $r$ endpoints---with a $\poly(r)$ increase in work, while retaining the $\poly(\log n)$ depth. 
\end{abstract}

\section{Introduction and related work}
There has been extensive work on \textit{(sequential) dynamic algorithms}, especially over the past decade. These algorithms adjust the output solution to small updates in the input with only a small amount of computational work per update---and crucially without resolving the problem from scratch. These dynamic algorithms are useful in two ways: (A) In \textit{practical settings that have an intrinsic dynamic nature}, where the input changes slowly over time and we should adapt to these changes. (B) In \textit{static algorithms} that indirectly induce dynamic problems: sometimes when we use one algorithmic subroutine as a part of another larger algorithm, the operations of the larger algorithm cause small changes to the problem solved by the subroutine. The computational efficiency of the overall algorithm then hinges on the subroutine being able to quickly process the updates with work proportional to only the number of updates, rather than the entire problem size. See for instance the recent breakthrough of an almost linear-time maximum flow algorithm~\cite{chen2022maximum}, which relies on efficient dynamic algorithms. See also \cite{hanauer2022recent} for more on dynamic algorithms, including a survey (which has an experimental focus).

In this paper, we seek parallel dynamic algorithms. The general motivation is similar to the sequential setup: such algorithms can be useful in practical dynamic scenarios while allowing us to benefit from parallelism, and they can also play a crucial role in achieving more efficient parallel algorithms for static problems. However, there is one key difference in what we desire from a dynamic algorithm, namely handling many updates at the same time. We discuss that after a quick reminder on the parallel model of computation.

\paragraph{Parallel computation}. We follow the standard work/depth terminology, see e.g., \cite[Section 1.5]{jaja1992introduction} or \cite{blelloch1996programming}. For any algorithm, its \textit{work} is defined as the total number of operations. Its \textit{depth} is defined as the longest chain of operations with sequential dependencies, in the sense that the $(i+1)^{th}$ operation depends on (and should wait for) the results of operation $i$ in the chain. These bounds can be translated to a running time for any given number of processors. With $p$ processors, an algorithm with work $W$ and depth $D$ can be run in $O(W/p+D)$ time~\cite{brent1974parallel}. Therefore, the ultimate goal in parallel algorithms is to devise algorithms with work asymptotically equal (or close, e.g., equal up to logarithmic factors) to their sequential counterpart---such an algorithm is called work efficient (respectively, nearly work efficient)--- which also have a small depth, ideally only polylogarithmic. Such algorithms enjoy optimal speed-up for any given number of processors, up to logarithmic factors.

\paragraph{Dynamic parallel algorithms}. A crucial difference when considering dynamic algorithms in the parallel setting with those in the sequential setting is the importance of processing \textit{many updates} simultaneously. In the sequential setting, we can afford to examine the updates one by one, and that makes the task much easier. In contrast, when we have many processors, there is no excuse to accept and process the updates only one by one. The distinction is even more important when building dynamic parallel algorithms that will be used in the context of a larger static parallel algorithm. If the parallel dynamic subroutine processes the updates one by one, the static parallel algorithm's \textit{depth} would be at least linear in the total number of updates it creates, which would usually make the latter uninteresting.~\footnote{We note that the literature sometimes uses the phrase \textit{batch-dynamic} parallel algorithms. We find it more natural to use the simpler phrase \textit{dynamic parallel algorithms}, because the need for processing many updates comes naturally when one considers parallelism. One can argue that a parallel dynamic algorithm that can process only a single update per time-step would have limited relevance and applicability (especially in problems where the sequential dynamic algorithm's work per update is very small).} 

We want parallel counterparts of the sequential dynamic algorithms, which can benefit from parallelism optimally up to logarithmic factors. In light of the above discussion, our concrete objective is parallel dynamic algorithms that can process updates with work-per-update similar to their sequential counterpart but crucially in only polylogarithmic depth for any batch of simultaneous updates. Note that, if we ignore logarithmic factors, such a result is stronger than (or simultaneously generalizes) two counterparts: (1) the sequential dynamic algorithm for that problem, and (2) the static parallel algorithm for that problem.

There has been growing recent interest in parallel dynamic algorithms. An interesting example is the work of Acar et al.~\cite{acar2019parallel} on parallel dynamic algorithms for graph connectivity, which provides a parallel counterpart of the sequential dynamic algorithm of Holm, de Lichtenberg, and Thorup~\cite{holm2001poly}. This example is notable because it has already found applications in improving the static parallel algorithms: Ghaffari, Grunau, and Qu~\cite{ghaffari2023nearly} used it to achieve the first work-efficient and sublinear-depth parallel DFS algorithm. See also the work of Tseng, Dhulipala, and Shun~\cite{tseng2022parallel} on parallel dynamic algorithms for minimum spanning forest, and that of Liu et al.~\cite{liu2022parallel} on parallel dynamic algorithms for k-core decomposition. 

\paragraph{Maximal matching, and sequential dynamic algorithms.} Our focus in this paper is on matching, one of the most basic and well-studied graph problems. Maximal matching, and its relatives including constant approximations of maximum matching, have received extensive attention in the sequential dynamic setting; see e.g., ~\cite{onak2010maintaining,baswana2011fully,bhattacharya2015deterministic, bhattacharya2016new,solomon2016fully,assadi2021fully,kiss2022deterministic,bhattacharya2023dynamicBeat2,bhattacharya2023dynamic1Eps, azarmehr2024fully}. The following results are of particular relevance for us\footnote{This summary discusses the most relevant prior work and ignores, for instance, algorithms with polynomial update time (e.g., for deterministic maximal matching). Also, our focus is maximal matching and we do not discuss improved approximations of maximum matching.}: Onak and Rubinfeld~\cite{onak2010maintaining} gave the first algorithm that maintains a constant approximate matching in $\poly(\log n)$ amortized time per update (edge deletion or insertion). This was improved to $O(\log n)$ time per update and for the more stringent maximal matching problem by Baswana et al.\cite{baswana2011fully}, then to constant amortized time per update for maximal matching by Solomon~\cite{solomon2016fully}. Assadi and Solomon gave a generalized result for maximal matching in hypergraphs~\cite{assadi2021fully}, with $O(r^2)$ amortized time in hypergraphs of rank $r$. All these algorithms are randomized and assume that the updates are determined by an \textit{oblivious adversary} that does not know the randomness used by the dynamic algorithm.

\paragraph{Our result---dynamic parallel maximal matching.}
To the best of our knowledge, there is no known parallel dynamic algorithm for maximal matching (or constant-approximate matching) with polylogarithmic amortized time per update in general graphs.\footnote{See the work of Liu et al.\cite{liu2022parallel}, which achieved this for low-arboricity graphs.} Indeed, prior work on experimental algorithmics by Angriman et al.\cite{angriman2022batch} pointed out the lack of, and the desire for, such an algorithm with provable guarantees. Their work provides a heuristic with good experimental results but without theoretical guarantees. 

We give a parallel dynamic algorithm that maintains a maximal matching with $\poly(\log n)$ amortized work per update, and it processes any arbitrary size batch of updates in $\poly(\log n)$ depth. These bounds assume we use the algorithm for only $\poly(n)$ updates--see the more general statement below. Up to logarithmic factors, this is an ideal parallelization of the sequential algorithms reviewed above. Similar to those sequential dynamic algorithms, our algorithm is also randomized and works against an adversary that chooses the updates obliviously to the randomness used by our algorithm. Moreover, our algorithm also generalizes to hypergraphs of rank $r$ with a $\poly(r)$ increase in the work bound.

\begin{theorem} [Imprecise statement; see \Cref{thm:main} for precise bounds]
    There is a randomized parallel dynamic algorithm for hypergraph maximal matching, which adapts the maximal matching to any batch of simultaneous edge insertion/deletion updates. The batch update is processed, with high probability, in $\poly(\log (nM))$ depth and using $poly(r\log (nM))$ amortized work per update. Here $r$ denotes the rank of the hypergraph, $n$ denotes the number of vertices, and $M$ denotes the total number of hyperedges inserted or deleted (starting from an empty graph).
\end{theorem}

We comment that, given the unspecified polylogarithmic factors, the statement applies to all variants of the PRAM model, e.g., including the weakest variant with exclusive reads and exclusive writes (EREW). The rest of this write-up will focus only on the work and depth of the algorithm, and will not discuss concurrency issues for a particular model variant, considering that those can be resolved with logarithmic overheads~\cite{jaja1992introduction}.

As another side note, we comment that our maximal matching algorithm can be adjusted to work in the distributed setting (in particular the standard \congest model~\cite{Peleg2000DistributedApproach}, where per round one $O(\log n)$-bit message can be sent along each edge): it would process each batch of updates in $\poly(\log n)$ rounds while communicating amortized $\poly(\log n)$ messages per update. We defer the details of this adaptation to the full version of our paper.

\section{Preliminaries}

\label[section]{section:Preliminaries}

\paragraph{Hypergraphs and maximal matching.}
We work with a hypergraph $H=(V, E)$ where $V$ denotes the set of vertices and $E$ denotes the set of hyperedges. Each hyperedge $e\in E$ is a subset of $V$, of cardinality at most $r$. We refer to $r$ as the rank of this hypergraph. Also, for a vertex $v\in V$ and an edge $e\in E$, the notation $v\in e$ means that $v$ is one of the at most $r$ endpoints of $e$. In the course of our algorithm, often, when it is clear from context, we will refer to the hypergraph and its hyperedges as simply the graph and its edges.

A matching $\mathcal{M}\subseteq E$ is a collection of hyperedges in $H$ that are disjoint, i.e., hyperedges that do not share an endpoint. The size of the matching is the number of hyperedges in it. A matching $\mathcal{M}$ in $H$ is maximal if there is no other matching $\mathcal{M}'$ in $H$ such that $\mathcal{M} \subset \mathcal{M}'$. Notice that any maximal matching has size at least $1/r$ factor of the maximum matching. 
A set $S\subset V$ of vertices is called a vertex cover if each edge in $E$ has at least one endpoint in $S$. For a maximal matching $\mathcal{M}$, the set of all endpoints of edges in $\mathcal{M}$ forms a vertex cover, and its size is at most $r$ times the minimum-size vertex cover. 

\paragraph{Update model.} The updates we consider on the hypergraph are edge insertions and deletions. We note that this is often called the fully dynamic model, to distinguish it from insertions-only or deletions-only setups. Notice also that each update requires $\Theta(r)$ words to describe the up to $r$ vertices of the hyperedge. Also, as mentioned before, our algorithm is randomized and works against an oblivious adversary (this is common also to the sequential dynamic algorithms for maximal matching mentioned before~\cite{baswana2011fully, solomon2016fully,assadi2021fully}). That is, the updates are determined by an adversary that can know our algorithm but does not know the randomness used by our algorithm. We consider that in each iteration, a batch of updates arrives, of an arbitrary size. We want the parallel algorithm to finish processing the entire batch in $\poly(\log n)$ depth, and to output all the changes to the maximal matching (it also updates internal data structures that it maintains).

\paragraph{Parallel dictionary.} To keep sets efficiently with parallel updates, we make use of a parallel dictionary data structure. Such a data structure supports batch insertions, batch deletion, and batch look-ups of elements from some universe with hashing. Gil et al.\cite{gil1991towards} describe a parallel dictionary that uses space linear in the number of items present in the dictionary, and for any batch of $k$ operates, performs $O(k)$ work and $O(\log^{*} k)$ depth, with probability $1-1/\poly(k)$. We sometimes will need a stronger probabilistic guarantee. For any $N \geq k$, we can get success probability $1-1/\poly(N)$ by increasing the work and depth bounds to $O(k \log N)$ work and $O(\log N)$, respectively. Moreover, since the space is linear in the number of items present in the data structure, we can also report all those items in depth $O(1)$ and work linear in the number of items.

\paragraph{Static Parallel Hypergraph Maximal Matching}
Sometimes in our algorithm, we need to build a maximal matching for some subsets of edges from scratch. For that, we use Luby's algorithm from \cite{luby}, for finding a maximal independent set in a graph. Note that finding maximal matching in a hypergraph $H=(V, E)$ can be reduced to finding a maximal independent set in the graph whose vertices are edges from $E$, and where any two of these vertices (former "hyperedges") $e$ and $e'$ are considered adjacent if they share a vertex in $H$.

\begin{center} \fbox{\begin{minipage}{0.97\textwidth}
\paragraph{Luby's MIS algorithm.} Given a graph on $n$ nodes, repeat the following process until there are no more vertices left. 

\begin{enumerate}
    \item Let $X_1, X_2, \ldots, X_n$ be random real numbers from range $[0, 1]$.
    \item Determine $I$: the set of nodes $v$, such that $X_v$ is larger than $X_u$ for any neighbor $u$ of $v$. 
    \item Add set $I$ to the independent set, and remove $I$ and all neighbors of $I$ from the graph.
\end{enumerate}  
\end{minipage}}\end{center}


\begin{theorem}[Luby~\cite{luby}]
 The process terminates in $O(\log{n})$ iterations with probability at least $1-1/\poly(n)$.
\end{theorem}


\begin{theorem} 
\label[theorem]{theorem:MaximalMatching}
There is a randomized parallel algorithm that, in any hypergraph of rank $r$, computes a maximal matching in depth $O(\log{M})$ and work $O(Mr \cdot \log{M})$, with high probability. Here, $M$ denotes the total number of hyperedges.
\end{theorem}

\begin{proof}
    We perform Luby's algorithm for the hyperedges of the graph. In each iteration, we generate a random number from $[0, 1]$ for each edge, then keep only local maximums as edges chosen to join the matching. We then remove all hyperedges incident to the chosen edges. Each iteration can be done in parallel with depth $O(1)$ and work $O(Mr)$. This process simulates finding the maximal independent set in a graph, whose vertices are hyperedges from $E$, and any two such vertices (former hyperedges) $e$ and $e'$ share an edge in the graph iff they shared a vertex in the hypergraph. Since this graph has $M$ vertices, the algorithm terminates in $O(\log{M})$ iterations, with high probability.
\end{proof}

\section{Algorithm}

\subsection{Intuitive discussions and algorithm overview}
Computing a (hypergraph) maximal matching in the sequential setting is trivial (even in time linear in the number of edges). However, that is not the case if we want a \textit{dynamic sequential algorithm} for it with $\poly(\log n)$ work per update, or if we want a low-depth and work-efficient \textit{static parallel algorithm} for it, i.e., with $\poly(\log n)$ depth, and $\poly(\log n)$ work per edge. The parallel dynamic algorithm that we desire in this paper is stronger than both of these---it would directly imply a sequential dynamic algorithm if we just run it on one processor, and it would directly imply the static parallel algorithm if we just insert all the edges at once. We borrow ideas from the known solutions to both of these problems, and we need some new ingredients in addition.

On a high level, our algorithm follows the general structure of the sequential dynamic algorithms for maximal matching~\cite{baswana2011fully,solomon2016fully,assadi2021fully}---notably the leveling scheme that is common to all of these. However, it also involves adaptations that enable us to process several updates simultaneously, in parallel. Our notation will be close to that of \cite{assadi2021fully}, to allow the generality for hyperedges; though in parts our algorithm behaves closer to that of ~\cite{baswana2011fully}, with new ingredients that give way for parallelism. 

The main challenge for any dynamic maximal matching algorithm is handling deletions of hyperedges. If the hyperedge $e$, which is to be deleted, is not in the current matching, the task is easy. But if $e$ is in the matching, we might need a lot of work: Naively, for each node $v \in e$, we should look for a hyperedge incident to $v$ such that all of its nodes are currently unmatched. If we find such a hyperedge, we add it to the matching. However, going through the entire incidence list of $v$ is expensive, as each node $v$ might have many incident edges (and we have to check for each endpoint of $e$). 

A basic idea is that if the matching edge $e$ was sampled uniformly at random from some large subset $S$ of hyperedges incident to $v$, then intuitively (and certainly relying on the oblivious adversarial model) we expect that around $\frac{|S|}{2}$ of edges of $S$ will get deleted before $e$. Since deleting these edges is easy, this should give us enough budget to amortize the expensive work of processing the deletion of the matching edge $e$. This idea is at the core of all prior work~\cite{baswana2011fully,solomon2016fully,assadi2021fully}. However, turning this raw intuition into an actual dynamic algorithm requires careful structures.

A key structure used by \cite{baswana2011fully,solomon2016fully,assadi2021fully}, which we will also utilize, is to keep vertices at different ``levels'' based on how many incident edges they had (the timing of this needs care) and thus intuitively how expensive it will be to process an incident matching edge deletion. More concretely, we use a \textit{leveling scheme} with the following intuitive properties (the precise definition is slightly different and is presented later). The edges of the matching are going to be placed at $L = \log{M}$ levels, where for an edge $e\in \mathcal{M}$ at level $\ell$, we have the following two properties: 

\begin{itemize}
    \item If $e$ is deleted, it would take the algorithm roughly $\tilde{\Theta}(r^{\ell})$ work to restore the maximality of the matching;
    \item $e$ was sampled from $\tilde{\Theta}(r^{\ell})$ edges, all of which were incident to the same vertex. The edges that weren't sampled don't violate the maximal matching, and in expectation roughly $\tilde{\Theta}(r^{\ell})$ edges get deleted before $e$. This would give us the budget for processing the deletion of the matching edge $e$ once that event occurs.
\end{itemize}
Moreover, due to the mechanics of how these levels are maintained, sometimes adding an edge to a matching might kick out other matching edges and this creates a cascade of work to be handled, and for which we need to have a budget of amortized work to pay.

\paragraph{Difficulties in parallelism, when handling multiple updates simultaneously.} The main challenge in adapting these general approaches to the parallel setting is the need to find new matching edges for many nodes at once if many matching edges get deleted. These matching edges cannot be arbitrary; they have to have enough properties to allow us to conclude a budget creation for each matching edge, similar to the above. Unlike the sequential setting, in which the sets from which each matching edge is sampled are disjoint, here these sets might heavily intersect with a complex structure. 

To remedy this problem, we will compute the matching edges in many iterations, each time making a small amount of progress. One ingredient is that we slightly relax the conditions and allow for smaller sizes of the sets from which matching edges were sampled. On a high level, we will argue that if a matching edge $e$ gets selected from set $S$ during a corresponding parallel process of matching edge selection, all edges incident to $e$ have comparable probabilities to be in the matching, where probabilities will be equal up to roughly $\poly(r)$ factors. As such, we will be able to still create enough budget to handle each expensive edge deletion. 

Furthermore, whether a node's level should be adjusted or not usually depends on the presumption that other nodes are staying at their current level. But when we want a parallel algorithm that processes many updates at the same time, this assumption will not be true and other nodes will try to move their levels at the same time. We will need mechanisms to ensure that the movements made by one node, under the presumption of others staying at their position, maintain the invariants approximately even when others are moving simultaneously. Finally, we will need extra details and careful arguments to show that we have probabilistic concentrations about these desired properties (rather than just the expectations, which are often considerably simpler). 

\subsection{General structure: levels, temporary deletions, and data structures}
This subsection describes the general structure that the algorithm preserves, the corresponding data structures, and the corresponding procedures we use for making changes and updating these data structures. The next subsection describes how our dynamic parallel algorithm makes use of these structures to process the updates to the hypergraph. As mentioned before, this overall outline builds on prior work~\cite{baswana2011fully,solomon2016fully,assadi2021fully}, and we keep our notation and terminologies mostly similar to those written in \cite{assadi2021fully}.

In general, our algorithm will maintain a maximal matching $\mathcal{M}$.  We will call hyperedges from $\mathcal{M}$ \textit{matched}; other hyperedges will be called \textit{unmatched}. We will call a node \textit{matched} if it is incident to any matched hyperedge, and \textit{unmatched} otherwise.

\subsubsection{A Leveling Scheme}
Lets us define $$\alpha = 4\cdot r, \qquad L = \lceil \log_{\alpha} N \rceil$$ 

Here, $N$ is a constant-approximate upper bound for the number of vertices plus the total number of updates (the total number of times edges were inserted/erased). Once $N$ more updates arrive, we can double $N$, rebuild all data structures from scratch, and charge all the previous work to the $N$ new updates. This way, every edge will be charged at most twice. We may therefore treat $N$ as a fixed value in the remainder of the paper.


\paragraph{Leveling scheme} We assign to each node $v \in V$ and each hyperedge $e \in E$ a level. We denote these by $\ell(v), \ell(e)$ correspondingly. These levels satisfy the following key invariant with respect to the maximal matching $\mathcal{M}$ that we maintain.

\begin{invariant} 

\label[theorem]{theorem:invariant1}

The leveling scheme will have the following properties:
    \begin{enumerate}
    \item For any hyperedge $e \in E$, we have $0 \leq \ell(e) \leq L$, and for any vertex $v\in V$, we have $-1 \leq \ell(v) \leq L$, such that $\ell(v) = -1$ if and only if $v$ is unmatched.
    \item For any matched hyperedge $e \in \mathcal{M}$ and any incident $v \in e$, $\ell(v) = \ell(e)$.
    \item For any unmatched hyperedge $e \notin \mathcal{M}$, we have $\ell(e) = \max_{v \in e} \ell(v)$.
\end{enumerate}
\end{invariant}

\paragraph{Hyperedge ownership} Each hyperedge will be \textit{owned} by one of its nodes with the highest level. If there are many such nodes, ties can be broken arbitrarily. We will denote the set of hyperedges owned by the node $v$ by $\mathcal{O}(v)$. 

\subsubsection{Temporarily Deleted Hyperedges}
Sometimes, we will consider some hyperedges \textit{temporarily deleted}. We will not keep these temporarily deleted hyperedges in any other data structures. We will guarantee the following invariant:

\begin{invariant}
Any hyperedge that is temporarily deleted is incident on a matched hyperedge.
\end{invariant}

Every matched hyperedge $e \in \mathcal{M}$ will maintain a set of deleted hyperedges $\mathcal{D}(e)$, for which it is responsible. When the matched edge $e$ gets deleted, we will simply insert all hyperedges from $\mathcal{D}(e)$ back to the graph.

\paragraph{Intuition Behind Temporarily Deleted Edges.} The set $\mathcal{D}(e)$ roughly represents the set of hyperedges from which a matching hyperedge $e$ was sampled. Thus, we expect that before $e$ gets deleted, some significant fraction of $\mathcal{D}(e)$ gets deleted. This would allow us to amortize the cost of processing the deletion of $e$ when that happens to all edges of $\mathcal{D}(e)$ that were deleted prior to that. It is important that the edges in $\mathcal{D}(e)$ do not participate in the rest of the algorithm before $e$ gets deleted. In particular, these edges cannot be in the set $\mathcal{D}(e')$ for any other matching edge $e'$ before $e$ gets deleted. This ensures that the ``budget'' coming from the deletion of edges in $\mathcal{D}(e)$ is not allocated multiple times.

\subsubsection{Data structures}
In the description of the data structures that follow, whenever we say set, we mean the parallel dictionary discussed in \cref{section:Preliminaries}. We are only going to use the following interface for such a set $S$:

\begin{itemize}
    \item \texttt{insert($S_{insert}$)}: given set $S_{insert}$ of elements, insert all of them to $S$. This operation has depth $O(\log{N})$ and work $O(S_{insert}\log{N})$ with high probability;
    \item \texttt{erase($S_{erase}$)}: given set $S_{erase}$ of nodes, each of which is currently in $S$, erase all of them from $S$. This operation has depth $O(\log{N})$ and work $O(S_{erase}\log{N})$ with high probability;
    \item \texttt{retrieve()}: return the set of all elements that are currently in $S$ (in practice, this would mean assigning each of them to a separate processor). This operation has depth $O(\log{N})$ and work $O(S\log{N})$ with high probability.
\end{itemize}

To keep the structures mentioned earlier, similar to  \cite{assadi2021fully}, we use the following data structures.

\paragraph{Data structures for vertices.}

For each vertex $v \in V$ we maintain the following information:

\begin{itemize}
\setlength\itemsep{0.5em}
    \item $\ell(v):$ the level of $v$ in the leveling scheme;
    \item $\mathcal{M}(v)$: the hyperedge in $\mathcal{M}$ incident on $v$ (if $v$ is unmatched $\mathcal{M}(v) = \perp$);
    \item $\mathcal{O}(v):$ the set of hyperedges $e$ owned by $v$; we define $o_v := |\mathcal{O}(v)|$;
    \item $\mathcal{N}(v):$ the set of hyperedges $e$ incident on $v$;
    \item $\mathcal{A}(v, \ell)$ for any integer $\ell \geq \ell(v)$: the set of hyperedges $e \in \mathcal{N}(v)$ that are \textit{not} owned by $v$ and have level $\ell(e) = \ell$; we define $a_{v, \ell} = |\mathcal{A}(v, \ell)|$;
    \item $\tilde{o}_{v, \ell}$ for every integer $\ell > \ell(v)$: the number of hyperedges $v$ will own \textit{if} we increase its level from $\ell(v)$ to $\ell$. Note that if we increase level of $v$ to $\ell$, then for any edge $e\in \mathcal{A}(v, \ell')$ with $\ell(v) \leq \ell' < \ell$, node $v$ will become the node with highest level in $e$, so $e$ would then be owned by $v$. This means that $\tilde{o}_{v, \ell} = (\sum_{\ell' = \ell(v)}^{\ell-1} a_{v, \ell'})+ o_v$. Another issue to point out is that this definition assumes $v$ moves up but all other nodes remain at their current level. In our parallel algorithms, many nodes would want to move up at the same time, and thus we need more care in using this definition. 
 \end{itemize}
\noindent We also introduce the following notation, but the corresponding data structure will \textit{not} be maintained explicitly:

    \begin{itemize}
        \item $\tilde{\mathcal{O}}_{v, \ell}$ for every integer $\ell > \ell(v)$: the set of hyperedges that $v$ will own \textit{if} we increase its level from $\ell(v)$ to $\ell$. In other words, $\tilde{O}_{v, \ell}$ is a union of $\mathcal{O}(v)$ and $\mathcal{A}(v, \ell')$ for $\ell(v) \leq \ell' < \ell$. Note that $\tilde{o}_{v, \ell} = |\tilde{\mathcal{O}}_{v, \ell}|$. We do \textbf{not} explicitly maintain these sets: when we need to access all edges of $\tilde{\mathcal{O}}_{v, \ell}$, we can do this through $\mathcal{O}(v)$ and $\mathcal{A}(v, \ell')$ for $\ell(v) \leq \ell' < \ell$.
    \end{itemize}

\paragraph{Data structures for edges.} For each hyperedge $e \in E$ we will maintain the following information:

\begin{itemize}
\setlength\itemsep{0.5em}
    \item $\ell(e):$ the level of $e$ in the leveling scheme;
    \item $O(e):$ the single vertex $v \in e$ that owns $e$, i.e., $e \in \mathcal{O}(v)$;
    \item $\mathcal{M}(e):$ a Boolean variable to indicate whether or not $e$ is matched;
    \item $\mathcal{D}(e):$ a set of deleted hyperedges for which it is responsible.
 \end{itemize}


\paragraph{Data structures for rising nodes per level.} We will also maintain the following data structure for each $\ell \in [0, L]$ about the nodes who may rise to level $\ell$. We note that this was not needed in the sequential algorithm \cite{assadi2021fully}, and we use it in our work as the set of nodes that need to rise in level changes in a more complex way when we process multiple such movements simultaneously.

\begin{itemize}
    \item $S_{\ell}:$ the set of nodes $v$ with with $\ell(v)<\ell$ and $\tilde{o}_{v, \ell} \geq \alpha^{\ell}$.
\end{itemize} 



\subsubsection{Procedures for changing levels and updating data structures}

Next, we introduce the key procedures for updating these data structures, when some nodes change their levels. We later describe how the dynamic algorithm performs the updates by invoking these procedures. We use superscript $*^{old}$ to denote a parameter or data structure $*$ before the update and $*^{new}$ to denote a parameter or data structure $*$ after the update. 

\begin{center} \fbox{\begin{minipage}{0.97\textwidth}
	\textbf{Procedure \texttt{set-owner($e, v$)}.} Given a hyperedge $e$ and vertex $v \in e$ where $\ell(v) = \max_{u \in e} \ell(u)$, the procedure sets the owner of edge $e$ to be node $v$, i.e., $O(e) = v$.
\end{minipage} } \end{center}

Implementing \texttt{set-owner($e, v$)} requires updating only the following: (1) sets $\mathcal{A}(v, \ell)$ for $v \in e$, (2) values of $\tilde{o}(v, \ell')$ for $v\in e, \ell'>\ell$, and (3) the corresponding sets $S_\ell'$. If we have to perform several such procedures in parallel, we would have to update all these data structures in parallel, and we will ensure the following guarantee for the related computational depth and work.

\begin{claim} 
    For any $x\geq 1$, processing $x$ such procedures in parallel takes $O(\log{N})$ depth and $O(xr\log{N})$ work.
\end{claim}

\begin{proof}
    Updating $\mathcal{A}(v, \ell)$ for $v \in e$ for all of them takes $O(\log{N})$ depth and $O(xr\log{N})$ work. Updating the corresponding $\tilde{o}(v, \ell')$ for $v\in e, \ell'>\ell$ takes $O(\log{L})$ depth and $O(xr\log{L})$ work with the parallel prefix sums algorithm from \cite{parallelprefixsum}. Updating the sets $S_{\ell'}$ takes $O(\log{N})$ depth and $O(xr\log{N})$ work.
\end{proof}

\begin{center} \fbox{\begin{minipage}{0.97\textwidth}
	\textbf{Procedure \texttt{set-level($v, \ell$)}.} Given a vertex $v\in V$ and integer $-1 \leq \ell \leq L$, this procedure sets $\ell(v) = \ell$.
\end{minipage} } \end{center}


Let us describe how we would handle a single such procedure. First, we determine ownership of all hyperedges $e\in \mathcal{O}^{old}(v)$, by finding $\arg \max_{w\in e} \ell(w)$ and applying corresponding \texttt{set-owner} procedures. Then we move $v$ to level $\ell$, and, if $\ell > \ell^{old}(v)$, we make $v$ owner of all hyperedges in $\mathcal{A}(v, \ell')$ for $\ell^{old}(v) \leq \ell' < \ell^{new}(v)$, by invoking \texttt{set-owner} procedures.

\begin{claim} 
    Processing $x$ \texttt{set-level} procedures in parallel takes $O(\log{N})$ depth and $\sum_{\text{relevant }v}O((|\mathcal{O}^{old}(v)| + |\mathcal{O}^{new}(v)|)r\log{N}) + O(Lx)$ work.
\end{claim}

\begin{proof}
    Determining ownership of each edge that's present in any of $\mathcal{O}^{old}(v)$ takes $O(1)$ depth and work equal to the sum of $O(|\mathcal{O}^{old}(v)|r)$ over all relevant $v$. Then we would have to perform \texttt{set-owner} procedures for those edges, as well as for edges in $\mathcal{A}(v, \ell')$ for $\ell^{old}(v) \leq \ell' < \ell^{new}(v)$ for all relevant $v$. That is $\sum_{\text{relevant }v} |\mathcal{O}^{old}(v)| + |\mathcal{O}^{new}(v)|$ \texttt{set-owner} procedures. 

    This would take $O(\log{N})$ depth and $\sum_{\text{relevant }v}O((|\mathcal{O}^{old}(v)| + |\mathcal{O}^{new}(v)|)r\log{N}) + O(Lx)$ work. (Here $Lx$ comes from the need to inspect each $\mathcal{A}(v, \ell)$ for each $v$).

\end{proof}




\subsection{The Update Algorithm}

Whenever we receive a batch of new updates, we categorize them into three groups:

\begin{enumerate}
    \item Deleting unmatched hyperedges (or temporarily deleted hyperedges)
    \item Deleting matched hyperedges
    \item Inserting hyperedges
\end{enumerate}

We deal with these three groups one by one.

\subsubsection{Deleting unmatched hyperedges}

This is the easiest case. We can simply delete them, and update the corresponding data structures. (When we delete a temporarily deleted hyperedge, all we have to do is to delete it from the corresponding set $\mathcal{D}(e)$). 

\subsubsection{Deleting matched hyperedges}

After some matched hyperedges are deleted, we have to ensure the maximality of our matching. Let us call nodes of deleted matched hyperedges \textit{undecided}. The set of undecided nodes might change throughout our handling of this update.

We will go over each level from $L$ to $0$, one by one. In the end, we will make sure that the current matching is maximal.

\begin{center} \fbox{\begin{minipage}{0.97\textwidth}
	\textbf{Procedure \texttt{process-level($\ell$)}.} Processes level $\ell$, guaranteeing the following invariants: 

\begin{invariant}\label{inv:level} After \texttt{process-level($\ell$)}:
    \begin{itemize}
    \item There are no \textit{undecided} nodes in level $\ell$ or above (that is, all nodes in levels $\ell$ or above are matched);
    \item There is no node $v$ with $\ell(v)<\ell$ such that $\tilde{o}(v, \ell) \geq \alpha^{\ell}$.
\end{itemize}
\end{invariant}

\end{minipage} } \end{center}

Procedure \texttt{process-level($\ell$)} consists of two steps, described below.

\paragraph{Step 1: ensuring the first part of the invariant.} Let $U$ be the set of all undecided nodes $v$ on this level. Let $U_{free}$ denote the set of hyperedges such that all their nodes are currently unmatched, and each of them is owned by some node from $U$. 

Now we will run the maximal matching algorithm from \cref{theorem:MaximalMatching} on set $U_{free}$. For each newly matched hyperedge, we will place it and all its nodes onto level $0$. For each unmatched undecided node from level $\ell$, we will place it onto level $-1$.

\paragraph{Step 2: ensuring the second part of the invariant.} We need to make sure that there are no nodes $v$ with $\ell(v) < \ell$ and $\tilde{o}_{v, \ell} \geq \alpha^{\ell}$. Informally, while there exists any such node $v$, we would want to raise it to level $\ell$. Before going into details of this step, we will describe how we would perform it in a sequential setting, mirroring the procedure from \cite{assadi2021fully}. 

\paragraph{Performing Step 2 in sequential setting}. We will handle such nodes one by one, while they exist, with the following procedure:

\begin{center} \fbox{\begin{minipage}{0.97\textwidth}
	\textbf{Procedure \texttt{random-settle($v, \ell$)}.}  Handles a given vertex $v$ with $\ell(v) < \ell$ and $\tilde{o}_{v, \ell} \geq \alpha^{\ell}$.

\end{minipage} } \end{center}

First, we perform \texttt{set-level($v, \ell)$}. Now, $|\mathcal{O}(v)| \geq \alpha^{\ell}$. Then, we choose a random hyperedge $e$ from $\mathcal{O}(v)$, set the level of each node $u \in e$ to level $\ell$, and add $e$ to the matching. If $u$ was already matched by some hyperedge $e'$, we will delete hyperedge $e'$ (and reinsert it back later). Note that all other nodes of $e'$ become \textit{undecided}, but their levels are smaller than $\ell$, so we will deal with them when processing the lower levels.

Finally, we add all other edges from $\mathcal{O}(v)$ to $\mathcal{D}(e)$, and (temporarily) delete them from the graph. 

\paragraph{Performing Step 2 in parallel setting}. Unfortunately, in the parallel setting, we cannot afford to deal with such nodes $v$ with $\tilde{o}_{v, \ell} \geq \alpha^{\ell}$ one by one. Somehow, we want to raise many nodes at a time to level $\ell$. 

Let us consider the set $B$ of all nodes $v$ with $\ell(v) < \ell$ and $\tilde{o}_{v, \ell} \geq \alpha^{\ell}$ (at the beginning it is equal to $S_{\ell}$, but it's going to be modified). For this set $B$, we are going to perform the process \texttt{grand-random-settle($B, \ell$)}. 

\begin{center} \fbox{\begin{minipage}{0.97\textwidth}
	\textbf{Procedure \texttt{grand-random-settle($B, \ell$)}.}  Handles given set of vertices $B$, so that in the end, for each vertex $v \in B$, either $\ell(v) = \ell$, or $\tilde{o}_{v, \ell}<\frac{\alpha^{\ell}}{2}$.
\end{minipage} } \end{center}

At the beginning of the procedure, consider the set $E' = \cup_{v \in B} \tilde{\mathcal{O}}_{v, \ell}$. For each edge $e\in E'$, choose uniformly at random a node $h(e)\in e$.  
Next, we are going to perform the sub-procedure, which we will call \texttt{grand-random-subsettle($B, \ell$)}, until the postcondition of the procedure \texttt{grand-random-settle} becomes satisfied. 

\paragraph{\texttt{grand-random-subsettle($B, \ell$)}.} This procedure consists of $2\log{\alpha}$ phases. The $i$-th phase consists of $O(\log{|E'|})$ iterations of the subsubprocedure, called \texttt{grand-random-subsubsettle($B, \ell, i$)}.

\paragraph{\texttt{grand-random-subsubsettle($B, \ell, i$)}.} 

\begin{enumerate}
    \item Let $p = \frac{2^i}{\alpha^{\ell+2}}$. For each edge in $E'$, mark it with probability $p$.
    \item Consider those marked edges in $E'$, which don't have incident marked edges in $E'$. Now we are going to lift all of them to level $\ell$, and add them to $\mathcal{M}$. Here is how we do it. 
    
    Consider any marked edge $e$ without incident marked edges. For each $u \in e$, perform \texttt{set-level$(u, \ell)$}, and then add $e$ to $\mathcal{M}$. Note that $u$ could have already been matched by some edge $\mathcal{M}^{old}(u)\neq \perp$. In that case, we will treat the hyperedge $\mathcal{M}^{old}(u)$ as deleted from the matching; all nodes of those hyperedges become undecided. At the end of the update step, we will reinsert $\mathcal{M}^{old}(u)$ and all edges from $\mathcal{D}(\mathcal{M}^{old}(u))$ back to the graph.

    In addition, for each non-marked edge $e'$, for which $f(e') \in e$, add $e'$ to $\mathcal{D}(e)$, and (temporarily) delete $e'$ from the graph. 

    \item Update set $B$ accordingly: keep only nodes $v$ for which we still have $\ell(v) < \ell$ and $\tilde{o}_{v, \ell} \geq \frac{\alpha^{\ell}}{2}$. Update set $E'$ accordingly: set $E' = \cup_{v \in B} \tilde{\mathcal{O}}_{v, \ell}$.
\end{enumerate}

\subsubsection{Inserting hyperedges}

When some hyperedges are inserted, first, for each of them, we check if all its nodes are unmatched. Let $S_{free}$ denote the set of hyperedges that have all of their nodes unmatched. Then, we run the maximal matching algorithm from \cref{theorem:MaximalMatching} on set $S_{free}$. Next, for each newly matched hyperedge, we will place it and all its nodes onto level $0$. 

Finally, for each added hyperedge $e$, we will find $v = \arg \max_{u \in e} \ell(u)$, and set $v$ as owner of $e$.

\paragraph{Reinserting temporarily deleted hyperedges.} Throughout our processing, we might have deleted some matched hyperedges: some - forced by updates, some - forced by our update algorithm. For each such hyperedge $e$, we are going to insert all hyperedges from $\mathcal{D}(e)$ back to the graph, together with all other insertion updates.

This finishes the description of the algorithm.


\section{Analysis of the algorithm}

Our main result is presented below:

\begin{theorem}\label{thm:main} 
Our parallel dynamic algorithm for hypergraph maximal matching has the following guarantees:
    \begin{enumerate}
        \item The depth of processing any batch of updates is $O(L \cdot \log{a} \cdot \log^3{N})$, with high probability. 
        \item The work spent on processing $t$ updates is $O(t\alpha^8L^2\log^2{\alpha}\log^7{N})$, with high probability.
    \end{enumerate}
\end{theorem}

To help readability, we prove this theorem in two parts: We prove the guarantee on the depth of our algorithm in \cref{theorem:DepthBound}, and the guarantee on the work of our algorithm in \cref{theorem:WorkBound}.

\subsection{Depth}

\begin{lemma}
\label[lemma]{lemma:SettleLikely}
    After \texttt{grand-random-subsettle($B, \ell$)} $B$ becomes empty with probability at least $\frac{1}{2}$.
\end{lemma}

\begin{proof}
    Consider the current edge set $E' = \cup_{v \in B} \tilde{\mathcal{O}}_{v, \ell}$. For $e\in E'$, define degree $d(e)$ of an edge as the number of edges in $E'$ incident to $e$. We are going to show that, with high probability (in terms of $|E'|)$, after the $i$-th phase, there are no edges of degree at least $\frac{\alpha^{\ell+2}}{2^i}$ left in $E'$.

    Note that this is deterministically true before the first phase: since the degree of any node $v$ to $E'$, or $\tilde{o}_{v, \ell+1}$, does not exceed $\alpha^{\ell+1}$, the degree of any edge in $E'$ cannot exceed $\alpha^{\ell+1} \cdot r < \alpha^{\ell+2}$.

    Now, suppose that before we begin phase $i$, the degrees of all edges don't exceed $\frac{\alpha^{\ell+2}}{2^i}$. Consider any edge $e$ with $d(e) \in [\frac{\alpha^{\ell+2}}{2^{i+1}}, \frac{\alpha^{\ell+2}}{2^i}]$. We will show that during every single iteration of the procedure \texttt{grand-random-subsubsettle($B, \ell, i$)}, there is at least a $\frac{1}{2e^2}$ chance that it gets removed. This would imply that, with high probability (in terms of $|E'|$), an edge $e$ with $d(e) \in [\frac{\alpha^{\ell+2}}{2^{i+1}}, \frac{\alpha^{\ell+2}}{2^{i}}]$ will either get removed from $E'$ in $O(\log{|E'|})$ iterations of the procedure \texttt{grand-random-subsubsettle($B, \ell, i$)}, or its degree will drop below $\frac{\alpha^{\ell+2}}{2^{i+1}}$. We would then be able to take union bound over all $2\log{\alpha}$ phases of \texttt{grand-random-subsettle($B, \ell$)}.


    Remember that each edge gets marked with probability $p = \frac{2^i}{\alpha^{\ell+2}}$. If $e$ has $k \in [\frac{\alpha^{\ell+2}}{2^{i+1}}, \frac{\alpha^{\ell+2}}{2^i}]$ incident edges, then the probability of $e$ being not marked, and exactly one of these incident edges being marked is:
    
    \begin{align*}
        (1-p)\cdot kp(1-p)^{k-1} &= kp(1-p)^{k} 
        \geq kpe^{-pk} \\
        &\geq \frac{\alpha^{\ell+2}}{2^{i+1}} \cdot \frac{2^i}{\alpha^{\ell+2}} \cdot e^{-\frac{\alpha^{\ell+2}}{2^i} \cdot \frac{2^i}{\alpha^{\ell+2}}} \\
        &= \frac{1}{2e}
    \end{align*}

    Conditional on this event happening, consider this single incident marked edge. The probability that it has no incident marked edges is at least

    $$(1-p)^{\frac{\alpha^{\ell+2}}{2^i}} \geq e^{-p \cdot \frac{\alpha^{\ell+2}}{2^i}} = \frac{1}{e}$$

    So, the probability that $e$ has some incident marked edge $e'$ that does not have incident marked edges is at least $\frac{1}{2e^2}$. 

\end{proof}

\begin{lemma}
\label[lemma]{lemma:SettleFast}
    Procedure \texttt{grand-random-settle($B, \ell$)} has depth of $O(\log{\alpha}\log^3{N})$, with high probability.
\end{lemma}

\begin{proof}
    \texttt{grand-random-settle($B, \ell$)} consists of performing \texttt{grand-random-subsettle($B, \ell$)} until $B$ becomes empty. \texttt{grand-random-subsettle($B, \ell$)} consists of $O(\log{\alpha})$ phases of $O(\log{|E'|})$ iterations of \texttt{grand-random-subsubsettle($B, \ell, i$)}. Note that    
    \texttt{grand-random-subsubsettle($B, \ell, i)$)} has depth $O(\log{N})$. So, if we repeat \texttt{grand-random-subsettle($B, \ell$)} $k$ times, the depth is $O(k\log{\alpha}\log{|E'|}\log{N}) = O(k\log{\alpha}\log^2{N})$. By \cref{lemma:SettleLikely}, the probability that $B$ becomes empty after performing \texttt{grand-random-subsettle($B, \ell$)} is at least $\frac{1}{2}$, so we indeed have $k = O(\log{N})$ with high probability.

\end{proof}

\begin{theorem}
\label[theorem]{theorem:DepthBound}
    The computational depth of the algorithm for processing any batch of updates is $O(L \cdot \log{\alpha}\cdot \log^3{N})$, with high probability.
\end{theorem}

\begin{proof}
Consider each update separately. Deleting unmatched hyperedges takes $O(\log{N})$ depth. Inserting edges takes $O(\log{N})$ depth.

Deleting matched hyperedges is a bit more complicated: we process everything level by level with \texttt{process-level($\ell$)}. For each level from $L$ to $0$, first, we run the maximal matching algorithm from \cref{theorem:MaximalMatching} in $O(\log{N})$ depth, then determine the set $B$ in $O(\log{N})$ depth, and then perform \texttt{grand-random-settle($B, \ell$)}, which takes $O(\log{\alpha}\log^3{N})$ depth, with high probability. Hence, the depth for a single update is $O(L\cdot \log{\alpha} \cdot \log^3{N})$ with high probability.
\end{proof}

\subsection{Work}

\subsubsection{Epochs}

We will use the notion of epochs from \cite{solomon2016fully}, \cite{assadi2021fully}.

\begin{definition}[Epoch] For any time $t$ and any hyperedge $e$ in the matching $\mathcal{M}$ at time $t$, the epoch of $e$ and $t$, denoted by $epoch(e, t)$, is the pair $(e, \{t'\})$ where $\{t'\}$ denotes the maximal continuous time period containing $t$ during which $e$ was always present in $\mathcal{M}$ (not even deleted temporarily during one step and inserted back at the same time step). 

We further define \underline{level} of $epoch(e, t)$ as the level $\ell(e)$ of $e$ during the time period of the epoch.

\end{definition}

We are going to charge all costs of computation to the epochs, or to the insert/erase updates.

\subsubsection{Charging Computation Costs}


We will go over all the work performed by our algorithm and show how we can charge it to some epochs, or insert/erase updates.

\paragraph{Deleting non-matched edges.} Deleting $x$ non-matched edges takes $O(xrL\log{N})$ work: we just need to update the relevant data structures for $O(xr)$ nodes. We can afford to charge each deletion update $O(rL\log{N})$ work.

\paragraph{Inserting edges.} Inserting $x$ edges takes $O(xrL\log{N})$ work: we perform the maximal matching algorithm from \cref{theorem:MaximalMatching} on some subset of them in work $O(xr\log{N})$ with high probability, and then update the relevant data structures for $O(xr)$ nodes in work $O(xrL\log{N})$. 

There are two types of edge insertions: 

\begin{enumerate}
    \item Edge insertions coming from the updates. We can charge each such update $O(rL\log{N})$;
    \item When an edge $e\in M$ gets deleted, we have to reinsert all edges from $\mathcal{D}(e)$ back to the graph. Note that if $e$ got created at level $\ell$, the size of $\mathcal{D}(e)$ cannot exceed $\alpha^{\ell+1}$, by the \cref{inv:level}. Therefore, we can charge the epoch corresponding to $e$ $O(rL\log{N})$ work for the reinsertion of each such temporarily deleted edge. This adds up to $\alpha^{\ell+1}O(rL\log{N}) = O(\alpha^{\ell+2}L\log{N})$.
\end{enumerate}

\paragraph{Dealing with undecided nodes in \texttt{process-level($\ell$)}.} If $U$ is the set of undecided nodes on level $\ell$, and $U_{free}$ is the set of hyperedges 
such that all their nodes are currently unmatched, and each of them is owned by some node from $U$, we run the maximal matching algorithm from \cref{theorem:MaximalMatching}. Note that $|U_{free}| \leq |U|\cdot \alpha^{\ell+1}$, so running this algorithm takes $O(|U_{free}|r\log{N}) = O(|U|\alpha^{\ell+2}\log{N})$ work. Additionally, we then spend $O(|U_{free}|rL\log{N}) = O(|U|\cdot \alpha^{\ell+2}L\log{N})$ for updating data structures. 

Note that a node becomes undecided only as a result of its matched edge $e$ being deleted. Hence, we can charge the epoch corresponding to $e$ $O(\alpha^{\ell+2}L\log{N})$ work for each its undecided node, which adds up to $O(\alpha^{\ell+2}L\log{N}) \cdot r = O(\alpha^{\ell+3}L\log{N})$.

\paragraph{Performing \texttt{grand-random-settle($B, \ell)$}.} First, we find set $B$ in $O(|B|\log{N})$ work. Then we perform \texttt{grand-random-subsettle($B, \ell$)} until $B$ becomes empty, where \texttt{grand-random-subsettle($B, \ell$)} consists of $O(\log{\alpha})$ phases of $O(\log{|E'|}) = O(\log{N})$ iterations of \texttt{grand-random-subsubsettle($B, \ell, i$)}, which takes $O(|E'|L\log{N})$ work.

According to \cref{lemma:SettleFast}, $B$ will become empty in $O(\log{N})$ repetitions of this process, with high probability. So, the computation cost corresponding to \texttt{grand-random-settle($B, \ell)$} is in $O(|E'|L\log{\alpha}\log^3{N}) = O(|B|\alpha^{\ell+1}L\log{\alpha}\log^3{N})$, with high probability. 

During this process, we add some matched edges onto level $\ell$. We want to charge this cost of $O(|B|\alpha^{\ell+1}L\log{\alpha}\log^3{N})$ to these newly created epochs. For that, we need to show that there is enough of them.

\begin{lemma}
    During \texttt{grand-random-settle($B, \ell)$}, at least $\frac{|B|}{\alpha^3}$ hyperedges are added to $\mathcal{M}$ on level $\ell$. This is true even deterministically, not only with high probability.
\end{lemma}

\begin{proof}
    Denote the number of these newly added edges by $k$. Then at least $|B|-rk$ nodes from $B$ have not moved to level $\ell$. For each such node $v$, we have $o^{old}_{v, \ell}\geq \alpha^{\ell}$, and $o^{new}_{v, \ell}\leq \frac{\alpha^{\ell}}{2}$. So, $\sum_{v\in B} \tilde{o}_{v, \ell}$ decreased by at least $(|B|-rk)\cdot \frac{\alpha^{\ell}}{2}$.
    
    The value of $o_{v, \ell}$ can decrease only in virtue of some node $u$ from an edge $e$ from $\tilde{\mathcal{O}}_{v, \ell}$ moving up to level $\ell$: now edge $e$ cannot be owned by a node with level $< \ell$. Note that $u$ is moving up to level $\ell$ could cause this change in at most $\alpha^{\ell+1}$ edges and each such edge can decrease the value of $o_{v, \ell}$ for at most $r$ nodes. 
    
    So, from one side, the sum of $o_{v, \ell}$ over nodes from $B$ which were not lifted to level $B$, decreased by at most $rk\alpha^{l+1}r$. From the other side, we know that it decreased by at least $(|B|-rk)\frac{\alpha^{\ell}}{2}$. So, we get: 

    $$(|B|-rk)\frac{\alpha^{\ell}}{2} \leq rk\alpha^{l+1}r \implies (|B|-rk)\leq k(2r^2\alpha) $$ from which we conclude $$ |B|\leq k(2r^2\alpha + r)<k\alpha^3 \implies k\geq \frac{|B|}{\alpha^3}$$
    
\end{proof}

We charge $O(\frac{|B|\alpha^{\ell+1}L\log{\alpha}\log^3{N}}{|B|\alpha^{-3}}) = O(\alpha^{\ell+4}L\log{a}\log^3{N})$ work to each of these newly matched hyperedges at level $\ell$.

Summing up: each deletion/insertion was charged $O(rL\log{N})$, and an epoch at level $\ell$ was charged $O(\alpha^{\ell+4}L\log{a}\log^3{N})$, with high probability.

\subsubsection{Natural and Induced Epochs}

This subsubsection almost entirely repeats the corresponding section in \cite{assadi2021fully}. First, let us assume that in the end all hyperedges get deleted. This at most doubles the number of updates, and won't affect the asymptotic of our computation.

Recall that an epoch is terminated when the corresponding hyperedge $e$ is deleted from the maximal matching $\mathcal{M}$. There are two types of hyperedge deletions from $\mathcal{M}$ in the algorithm: the
ones that are a result of the adversary deleting a hyperedge from the graph, and the ones that are the result of the update algorithm to remove a matched hyperedge in favor of another (so the original hyperedge is still part of the hypergraph). Based on this, we differentiate between epochs as follows:

\begin{itemize}
    \item \textbf{Natural epoch}: We say an $epoch(e, *)$ is natural if it ends with the adversary deleting the hyperedge $e$ from $G$;
    \item \textbf{Induced epoch}: We say $epoch(e, *)$ is induced if it ends with the update algorithm removing the hyperedge $e$ from $\mathcal{M}$ (so $e$ is still part of $G$).
\end{itemize}

We will also need the notion of \textbf{induced} and \textbf{natural} levels. We say that a level $\ell$ is an \textbf{induced level} (respectively, \textbf{natural level}) if the number of induced level-$\ell$ epochs is greater than (resp., at most) the number of natural level-$\ell$ epochs.

Next, we will charge the computation costs incurred by any induced level to the computation costs at higher levels, so that in the end, the entire cost of the algorithm will be charged to natural levels. Specifically, in any induced level $\ell$, we define a one-to-one mapping from the natural to the induced epochs. For each induced epoch, at most one natural epoch (at the same level) is mapped to it; any natural epoch that is mapped to an induced epoch is called \textit{semi-natural}. For any induced
level $\ell$, all the natural $\ell$-level epochs are semi-natural by definition. For any natural level, all the natural epochs terminated at that level remain as before; these epochs are called \textit{fully-natural}.

Now, we are going to redistribute all charges between the fully-natural epochs, as follows:

\begin{enumerate}
    \item Consider any induced epoch $(e, *)$. Edge $e$ could get deleted only during the creation of some epoch at a higher level during the \texttt{grand-random-settle} process: when epoch $e'$ gets created, then, for each $v\in e'$, if $v$ was previously matched, that edge gets deleted. Note that creating epoch $e'$ leads to the deletion of at most $r$ induced epochs. We will charge the (recursive) work charged to each such induced epoch $e$ to $e'$ instead.
    \item Consider any semi-natural epoch $(e, *)$. If it was mapped to induced epoch $(e', *)$ at the same level, then we will charge the (recursive) work charged to $(e, *)$ to $(e', *)$ instead.
\end{enumerate}

\begin{lemma}
    For any $\ell \geq 0$, the (recursive) work charged to any level-$\ell$ epoch is bounded by $O(\alpha^{\ell+4}L\log{a}\log^3{N})$, with high probability.
\end{lemma}

\begin{proof}
    Denote by $\hat{C}(\ell)$ the (recursive) work charged to epoch at level $\ell$. We get the following recurrence:

    $$\hat{C}_{\ell} \leq 2\cdot (O(\alpha^{\ell+4}L\log{a}\log^3{N}) + r\hat{C}_{\ell-1})$$
    
    This recurrence solves to $\hat{C}_{\ell} = O(\alpha^{\ell+4}L\log{a}\log^3{N})$.
\end{proof}

\subsubsection{Bounding Computation Costs} In this section, we are going to show, that for the majority of epochs $(e, *)$ at any level $\ell$, a significant number of edges in $\mathcal{D}(e)$ gets deleted before $e$. This would allow us to bound the number of epochs created at level $\ell$.

Consider the sequence of deletions of all hyperedges. Note that it is not affected by our algorithm. Let us define the rank $r(e)$ of an epoch as the rank of an edge $e$ in this order. For an epoch $(e, *)$ at level $\ell$, consider $\mathcal{D}(e)$ at the time of the creation of this epoch, and ranks of all edges in $\mathcal{D}(e)$. 

\begin{definition}
    The uninterrupted duration of an epoch $(e, *)$ is defined as the number of edges in $\mathcal{D}(e) \cup \{e\}$, whose rank does not exceed $r(e)$.
\end{definition}

\begin{definition}
    An epoch $(e, *)$ at level $\ell$ is called $\mu$-short if its uninterrupted duration does not exceed $\mu \cdot \alpha^{\ell}$, for some parameter $0 \leq \mu \leq 1$.
\end{definition}

Our main goal is to show that for some $C = \frac{1}{\poly(\log{N})\poly(\alpha)}$, the fraction of epochs that are $C$-short is fairly small. For that, let us focus on \texttt{grand-random-settle$(B, \ell)$} again, and see how many short epochs are created there.

Consider any iteration of  \texttt{grand-random-subsubsettle($B, \ell, i$)} of \texttt{grand-random-settle$(B, \ell)$}, operating on edge set $E'$. 

\begin{definition}
    An edge $e\in E'$ is called \textit{$\mu$-almost short}, if there are at least $\mu \cdot \alpha^{\ell}$ edges in $E'$, incident to $e$, with rank not exceeding $r(e)$.
\end{definition}

\begin{lemma}
\label[lemma]{lemma:NotShort}

    If an edge $e$ that is \textbf{not} $\mu$-almost short is added to the matching on level $\ell$, it does \textbf{not} form an $\frac{\mu}{rO(\log{N})}$-short epoch, with high probability.
\end{lemma}

\begin{proof}
    Let $k > \mu\cdot \alpha^{\ell}$ be the number of incident edges to $e$ in $E'$. Each of them picked $e$ with probability at least $\frac{1}{r}$, independently of each other. 

    If $\frac{\mu}{rC\log{N}} \cdot \alpha^{\ell} \leq 1$ for some constant $C$, then the statement is trivially true since the uninterrupted duration of an epoch is always at least $1$ by definition. So, we can assume that $k \geq \mu\cdot \alpha^{\ell} > rC\log{N}$. Let us bound the probability that at most $\frac{k}{rC\log{N}}$ of these $k$ edges select $e$. If $X$ is the number of edges that select $e$, then $\mathbb{E}[X] \geq \frac{k}{r} \geq C\log{N}$, and for the right choice of $C$ we can apply Chernoff bound:

    \begin{align*}
    \mathbb{P}[X \leq \frac{k}{rC\log{N}}] &\leq 
    \mathbb{P}[X \leq \mathbb{E}[X](1 - \frac{C\log{N}-1}{C\log{N}})] \\
    &\leq e^{-\frac{C\log{N}}{2} \cdot (\frac{C\log{N}-1}{C\log{N}})^2} \leq N^{-\frac{C}{3}}    
    \end{align*}

    So, with probability at least $1 - N^{-\frac{C}{3}}$, at least $\frac{k}{rC\log{N}} \geq \frac{\mu\cdot \alpha^{\ell}}{rC\log{N}} = \frac{\mu}{rC\log{N}}\cdot \alpha^{\ell}$ will get to $\mathcal{D}(e)$, as desired.
\end{proof}

Before we go to the main lemma, we need to show one more probabilistic concentration statement.

\begin{lemma}
\label[lemma]{lemma:Gaming}
    Imagine that you are playing against an adversary. In their turn, they can give you some number of edges (they can also stop giving you edges, ending the game). In your turn, you mark each of the given edges with probability $p$. Given that the total number of edges the adversary gives you is $k = \poly(N)$, with high probability, the number of marked edges does not exceed $(kp + 1)O(\log{N})$.
\end{lemma}

\begin{proof}
    It might be tempting to apply Chernoff bound directly to all these $k$ edges, but note that we do not have independence here: the number of edges the adversary gives you on the next iteration might depend on the number of marked edges on the previous one. However, we still can provide a tight concentration bound. 

    We can imagine this as follows: there are $\poly(N)$ edges in a row, each of which is marked with probability $p$. The adversary does not see which edges are marked. In their turn, they can choose any number $x$, take the next $x$ edges, and reveal which of them are marked. We claim that even in this game, for any adversarial play, the following claim is true: if they reveal $k$ edges in total, the number of marked edges among them does not exceed $(kp+1)O(\log{N})$, with high probability.

    We can do this by simply applying the union bound over each $k$. Let $X_k$ denote the number of marked edges among the first $k$ edges in this row. $\mathbb{E}[X_k] = pk$. We need to bound $\mathbb{P}[X_k > (kp+1)\cdot C\log{N}]$, for some constant $C$.

    If $k \leq (kp+1) \cdot C\log{N}$, the probability is simply $0$. Otherwise, denote $\sigma = C\log{N} + \frac{C\log{N}}{kp}-1$, and write: 

       \begin{align*}\mathbb{P}[X_k>(kp+1)\cdot C\log{N}] &= \mathbb{P}[X_k > (kp)(1 + \delta)] \\
       &\leq e^{-(kp\sigma) \frac{\sigma}{\sigma+2}} \\
       &= e^{-(kp(C\log{N} - 1) + C\log{N})\frac{\sigma}{\sigma+2}}
       \end{align*}

    Since $\sigma, C\log{N} > 1$ for the large enough choice of $C$, we get:

    \begin{align*}\mathbb{P}[X_k>(kp+1)\cdot C\log{N}] &\leq e^{-(kp(C\log{N} - 1) + C\log{N})\frac{\sigma}{\sigma+2}} \\
    &\leq e^{-\frac{C}{3} \log{N}} = N^{-\frac{C}{3}}\end{align*}

    Since we need to consider only $k \leq \poly{N}$, by the union bound, with high probability, the statement holds for each relevant $k$ (under the right choice of constant $C$). 

\end{proof}

\begin{lemma}
    Let $T_{\ell}$ denote the total number of all epochs created on level $\ell \geq 1$ throughout the algorithm. With high probability, at most $\frac{T_{\ell}}{4} + O(\log{N})$ of them were $\mu$-almost short, where $\mu = 
 \frac{1}{O(\alpha^3\log{\alpha}\log^3{N})}$.  
\end{lemma}

\begin{proof}
    Consider an iteration of  \texttt{grand-random-subsubsettle($B, \ell, i$)} of \texttt{grand-random-settle$(B, \ell)$}, operating on edge set $E'$. There are at most $|B|\cdot \mu \cdot \alpha^{\ell}$ $\mu$-almost short edges in $E'$: at most $\mu \cdot \alpha^{\ell}$ edges with smallest rank in $\tilde{\mathcal{O}}_{v, \ell}$ for each $v\in B$ can be $\mu$-almost short. Each of them becomes marked with probability at most $\frac{1}{\alpha^{\ell}}$, and therefore gets matched on level $\ell$ with probability at most $\frac{1}{\alpha^{\ell}}$.

    Even though these markings are not independent, they behave according to the process described in the \cref{lemma:Gaming}. The total number of marked edges over all \texttt{grand-random-subsubsettle($B, \ell, i$)} procedures is trivially $O(\poly{N})$. So, if $X_{tot}$ is the total number of  almost short edges at level $\ell$ over all \texttt{grand-random-subsubsettle($B, \ell, i$)} procedures, the number of $\mu$-almost-short of them does not exceed $(\frac{X_{tot}}{\alpha^{\ell}} + 1)O(\log{N})$, with high probability, by \cref{lemma:Gaming}.

    We now consider the terms in $X_{tot}$. \texttt{grand-random-settle$(B, \ell)$} has $O(\log{\alpha}\log^2{N})$ iterations of \texttt{grand-random-subsubsettle($B, \ell, i$)} with high probability, so it contributes $O((|B|\mu\alpha^{\ell})\cdot (\log{\alpha}\log^2{N}))$ to $X_{tot}$ with high probability. From other side, it matches at least $\frac{|B|}{\alpha^3}$ edges to level $\ell$.

    Now, let $B_{tot}$ be the total sum of $|B|$ on level $\ell$ over all invocations of \texttt{grand-random-settle$(B, \ell)$}. Then $T_{\ell} \geq \frac{B_{tot}}{\alpha^3}$, and the number of $\mu$-almost short among them does not exceed $(\frac{B_{tot}\cdot \mu \cdot \alpha^{\ell}\cdot \log{\alpha}\log^2{N}}{\alpha^{\ell}} + 1)C\log{N}$ w.h.p. For our inequality, it is enough to show that w.h.p.

    $$B_{tot}\cdot \mu \cdot \log{\alpha}\log^2{N} \cdot C\log{N} \leq \frac{B_{tot}}{4\alpha^3}$$ which is equivalent to $$  \mu \leq \frac{1}{4\alpha^3 \cdot \log{\alpha}\log^2{N} \cdot C\log{N}} = \frac{1}{O(\alpha^3\log{\alpha}\log^3{N})}$$

So our choice of $\mu$ suffices.
\end{proof}


\begin{lemma}

\label[lemma]{lemma:FewShort}

    With high probability, the following statement holds for each level $\ell$:

    \begin{itemize}
        \item Let $T_{\ell}$ be the total number of all epochs ever created at level $\ell$. At most $\frac{T_{\ell}}{4} + O(\log{N})$ of them are $\mu$-short, where $\mu = \frac{1}{O(\alpha^4\log{\alpha}\log^4{N})}$.
    \end{itemize}
\end{lemma}

\begin{proof}
    First, focus on a single level $\ell$. We know that at most $\frac{T_{\ell}}{4} + O(\log{N})$ of epochs created on level $\ell$ were $\frac{1}{O(\alpha^3\log{\alpha}\log^3{N})}$-almost short, w.h.p.. By invoking \cref{lemma:NotShort}, we conclude that all epochs that were not $\frac{1}{O(\alpha^3\log{\alpha}\log^3{N})}$-almost short, did not become $\frac{1}{O(\alpha^3\log{\alpha}\log^3{N}) \cdot rO(\log{N}}) = \frac{1}{O(\alpha^4\log{\alpha}\log^4{N})}$ short, w.h.p. The lemma then follows from a union bound over all levels $\ell$.
\end{proof}

Assuming this \cref{lemma:FewShort}, we can bound the total computation cost, with high probability.

\begin{lemma}
    The total work charged to all epochs is at most  $t \cdot O(\alpha^8L^2\log^2{\alpha}\log^7{N})$, with high probability.
\end{lemma}

\begin{proof}
    Remember that we are charging only fully-natural epochs at natural levels. Let $T_{\ell}$ be the total number of all epochs ever created at level $\ell$, and $T'_{\ell}$ be the total number of $\mu$-short ones among them, where $\mu = \frac{1}{O(\alpha^4\log{\alpha}\log^4{N})}$. Then the number of fully natural not $\mu$-short epochs at level $\ell$ is at least $\frac{T_{\ell}}{2}-T'_{\ell}$.

    For each such epoch $(e, *)$, it gets deleted after at least $\mu\cdot \alpha^{l}$ deletions from $\mathcal{D}(e)$. Therefore, if $t$ is the total number of all updates, we get:
    $$(\frac{T_{\ell}}{2}-T'_{\ell})\mu\cdot \alpha^{l} \leq t $$ which implies $$ (\frac{T_{\ell}}{2}-(\frac{T_{\ell}}{4} + O(\log{N})))\mu\cdot \alpha^{l} \leq t$$ and thus we conclude $$ T_{\ell} \leq 4(\frac{t}{\mu\cdot \alpha^{l}} + O(\log{N})).$$

    \noindent On the other hand, the total work charged does not exceed 
    \begin{align*} 
    &  &&T_{\ell}\cdot O(\alpha^{\ell+4}L\log{\alpha}\log^3{N}) \\ &\leq &&4(\frac{t}{\mu\cdot \alpha^{l}} + O(\log{N}))\cdot O(\alpha^{\ell+4}L\log{\alpha}\log^3{N}) \\
    &= &&t \cdot O(\frac{\alpha^4L\log{\alpha}\log^3{N}}{\mu}) + O(\alpha^{\ell+4}L\log{\alpha}\log^4{N})\\
    &= &&t \cdot O(\alpha^8L\log^2{\alpha}\log^7{N}) + O(\alpha^{\ell+4}L\log{\alpha}\log^4{N})
\end{align*}
    
    Summing this up over all epochs, with high probability, the total work charged is upper bounded by the following:

    \begin{align*}
     & &&\sum_{l = 0}^{L} t\cdot O(\alpha^8L\log^2{\alpha}\log^7{N}) + O(\alpha^{\ell+4}L\log{\alpha}\log^4{N}) \\ 
     &=  && t\cdot O(\alpha^8L^2\log^2{\alpha}\log^7{N}) + O(\alpha^{\ell+4}L\log{\alpha}\log^4{N})\\  
     &= && t\cdot O(\alpha^8L^2\log^2{\alpha}\log^7{N})   
    \end{align*}
\end{proof}

\begin{theorem}
\label[theorem]{theorem:WorkBound}
    The total work performed by the algorithm is $tO(\alpha^8L^2\log^2{\alpha}\log^7{N})$, with high probability. Therefore, the amortized cost per update is $O(\alpha^8L^2\log^2{\alpha}\log^7{N})$, with high probability.
\end{theorem}

\begin{proof}
    Each deletion or insertion was charged $O(rL\log{N})$, and the total cost charged to all other epochs is $tO(\alpha^8L^2\log^2{\alpha}\log^7{N})$. So, we can upper bound the amortized computation cost per update by $O(\alpha^8L^2\log^2{\alpha}\log^7{N})$.
\end{proof}

\section{Conclusion and Open Problems} 
This paper presented a randomized parallel dynamic algorithm for maximal matching, with $\poly(\log n)$ depth to process an arbitrary-size batch of edge insertion/deletion updates, and with $\poly(\log n)$ amortized work per update. The algorithm also generalizes to hypergraphs of rank $r$ with a 
$\poly(r)$ factor increase in work per update. We view this result as a starting step in understanding the parallelism of dynamic algorithms for (maximal) matching. Contrast this with the sequential variant, for which we have learned gradually better results through a long line of improvements~\cite{onak2010maintaining,baswana2011fully,bhattacharya2015deterministic, bhattacharya2016new,solomon2016fully,assadi2021fully,kiss2022deterministic,bhattacharya2023dynamicBeat2,bhattacharya2023dynamic1Eps, azarmehr2024fully}. There are improved or alternative parallel dynamic results that one can desire. Two questions particularly stand out:

\begin{itemize}
    \item[(1)] Our algorithm is randomized and works against an oblivious adversary. This is a drawback that is also common to the prior work on sequential dynamic algorithms for maximal matching~\cite{baswana2011fully, solomon2016fully, assadi2021fully}. However, for the more relaxed problem of $(2+\eps)$-approximate maximum matching, for any constant $\eps>0$, there is a deterministic dynamic algorithm with $\poly(\log n)$ amortized work per update by Bhattacharya, Henzinger, and Nanongkai~\cite{bhattacharya2016new}. Can we get a similar result in parallel, i.e., a \textit{deterministic} parallel dynamic algorithm (which can work against an adaptive adversary) for constant-approximate maximum matching with $\poly(\log n)$ depth for processing any batch of updates and $\poly(\log n)$ amortized work per update? 
    
    \item[(2)] Our algorithm's work per update is $\poly(\log n)$. This is roughly similar to the work bound achieved by Onak and Rubinfeld~\cite{onak2010maintaining} and Baswana, Gupta, and Sen~\cite{baswana2011fully}. However, sequential dynamic algorithms have been further optimized and now algorithms with constant work per update are known, e.g., by Solomon~\cite{solomon2016fully} and with generalization to hypergraphs by Assadi and Solomon~\cite{assadi2021fully}. Can we get a parallel dynamic algorithm for maximal matching with $O(1)$ amortized work per update and $\poly(\log n)$ depth for processing any batch of updates?
    
\end{itemize}

\section*{Acknowledgement} This work was partially supported by a grant from the Swiss National Science Foundation (project grant 200021$\_$184735). 


\bibliographystyle{alpha}
\bibliography{references, ref2, ref3}

\end{document}